%% file: main.tex
\documentclass[10pt,journal,comsoc]{IEEEtran}

\usepackage[numbers,sort&compress]{natbib}

\ifCLASSINFOpdf
   \usepackage[pdftex]{graphicx}
\else
   \usepackage[dvips]{graphicx}
\fi
\graphicspath{{./plots/}{./figs/}}
\usepackage{amsmath}
\usepackage{amsmath,bm}
\usepackage{subfig}
\usepackage[normalem]{ulem}

\usepackage{xcolor}
\usepackage{booktabs}
\usepackage{multirow}
\usepackage{paralist}
\usepackage{amsfonts}
\usepackage{dsfont}
\usepackage{mathtools}
\usepackage[inline]{enumitem}
\usepackage[utf8]{inputenc}
\usepackage{tabularx}
\usepackage{changepage}
\usepackage{makecell}

\usepackage{amsthm}
\usepackage{algorithm}
\usepackage[noend]{algpseudocode}

\newtheorem{theorem}{Theorem}[section]

\newtheorem{lemma}[theorem]{Lemma}

\theoremstyle{definition}

\theoremstyle{remark}
\newtheorem{remark}{Remark}

\theoremstyle{remark}

\theoremstyle{remark}
\newtheorem{challenge}{Challenge}

\usepackage{url}
\newcommand{\red}[1] {\textcolor{red}{#1}}
\newcommand{\MA}[1] {\textcolor{green}{MA: #1}}

\usepackage{ulem}

\newcommand{\blue}[1] {\textcolor{blue}{#1}}

\newcommand{\ie}{\textit{i}.\textit{e}.}
\newcommand{\eg}{\text{e}.\textit{g}.}
\newcommand{\yes}{\checkmark}
\newcommand{\no}{$\times$}

\begin{document}
 \bstctlcite{IEEEexample:BSTcontrol}
\title{Optimal Service Placement, Request Routing and CPU Sizing in Cooperative Mobile Edge Computing Networks for Delay-Sensitive Applications}

\author{\IEEEauthorblockN{
        Naeimeh Omidvar\IEEEauthorrefmark{1}, Mahdieh Ahmadi\IEEEauthorrefmark{2}, and Seyed Mohammad Hosseini\IEEEauthorrefmark{3},
    }\\
    \IEEEauthorblockA{\IEEEauthorrefmark{1} School of Computer Science, Institute for Research in Fundamental Sciences (IPM), Tehran, Iran.}\\
    \IEEEauthorblockA{\IEEEauthorrefmark{2} University of Waterloo, ON, Canada.} \\
    \IEEEauthorblockA{\IEEEauthorrefmark{3} Sharif University of Technology, Tehran, Iran. }
    }

\IEEEtitleabstractindextext{%
\begin{abstract}

In this paper, we study joint optimization of service placement, request routing, and CPU sizing in a cooperative mobile edge computing (MEC) system. The problem is considered from the perspective of the service provider (SP), which delivers heterogeneous MEC-enabled delay-sensitive services, and needs to pay for the used resources to the mobile network operators and the cloud provider, while earning revenue from the served requests. We formulate the problem of maximizing the SP's total profit subject to the computation, storage, and communication constraints of each edge node and end-to-end delay requirements of the services as a mixed-integer non-convex optimization problem, and prove it to be NP-hard. To tackle the challenges in solving the problem, we first introduce a design trade-off parameter for different delay requirements of each service, which maintains flexibility in prioritizing them, and transform the original  optimization problem by the new delay constraints. Then, by exploiting a hidden convexity, we reformulate these delay constraints 
into an equivalent form. Next, to handle the challenge of the complicating (integer) variables, using primal decomposition, we decompose the problem into an equivalent form of master and inner sub-problems over the mixed and real variables, respectively. We then employ a cutting-plane approach for building up adequate representations of the extremal value of the inner problem as a function of the complicating variables and the set of values of the complicating  variables for which the inner problem is feasible. Finally, we propose a solution strategy based on generalized Benders decomposition and prove its convergence to the optimal solution within a limited number of iterations. Extensive simulation results demonstrate that the proposed scheme significantly outperforms the existing mechanisms in terms of various metrics, including the SP's profit, cache hit ratio, running time, and end-to-end delay.  

\end{abstract}

}

\maketitle
\IEEEdisplaynontitleabstractindextext  
\IEEEpeerreviewmaketitle

\input{part1}
\input{part2}
\input{part3}
\input{part4_v2}
\input{part5}
\section{Conclusions}\label{sec:conclusion}

In this paper, 
we studied the problem of maximizing the SP's total profit subject to the computation, storage, and communication constraints of each edge node and end-to-end delay requirements of the services, which is formulated as a challenging mixed-integer non-convex optimization problem. To tackle the challenges in solving the problem, we first introduce a design trade-off parameter for different delay requirements of each service, which maintains degrees-of-freedom in prioritizing them, and transform the original  optimization problem by the new delay constraints. Then, by exploiting a hidden convexity, we reformulate these delay constraints. 
To tackle the challenges caused by the complicating variables in solving the reformulated problem, first 
using primal decomposition, we decomposed the mixed integer problem into an equivalent form of inner and master sub-problems, and then employed a cutting-plane approach for building up adequate representations of the extremal value of the inner problem as a function of the complicating variables and the set of values of the complicating  variables for which the inner problem is feasible. Finally, based on generalized Benders decomposition, we proposed 
an iterative solution strategy and proved its convergence to the optimal solution within a limited number of iterations of the master problem. Extensive simulation results demonstrate that the proposed scheme significantly outperforms the existing mechanisms in terms of various metrics, including the SP's profit, cache hit ratio, running time, and end-to-end delay.


\bibliographystyle{IEEEtranN}
\bibliography{IEEEabrv,refs}

\end{document}

%% file: part1.tex
\section{Introdunction}

\IEEEPARstart{T}{he}  escalating demand for computation-intensive and delay-sensitive applications in next-generation wireless communication networks poses significant challenges for service providers (SPs) 
in effectively meeting user demands within a centralized cloud infrastructure. 
Mobile Edge Computing (MEC) emerges as a solution which introduces a distributed cloud architecture, bringing computing capabilities closer to end-users, 
i.e., to the network edge \cite{mach2017mobile}. This proximity allows users to overcome the inherent latency limitations of prevalent centralized cloud systems, ensuring the fulfilment of the service  requirements. 
As such, MEC is identified as a prominent approach that 
will serve a crucial role in next-generation networks 
\cite{hu2015mobile}. 
However, the limited resources at edge nodes (especially compared to the centralized clouds), particularly in terms of computation and storage, create major challenges 
for SPs for effective service provision. 
Consequently, \textit{service placement} becomes crucial in MEC networks, determining which services should be cached at each Base Station (BS) to better  utilize  the limited resources at the edge nodes. 

To address storage and computation challenges, \textit{cooperative MEC} introduces collaborative resource sharing and task management among edge nodes, aiming to enhance efficiency, reliability, and responsiveness in MEC networks. 
However, realizing the performance enhancements and improved capabilities of cooperative MEC requires 
rigorous management of coordination and collaboration among multiple edge nodes. 
In particular, under the cooperative approach, a service request received at a BS may need to be 
delegated  to a neighbouring BS  which has cached such service and has also spare available resources. Therefore, the problem of \textit{request routing} 
is another important problem which further complicates cooperative MEC schemes. 
Furthermore, apart from intensive storage and computation requirements, many modern services also require uploading and routing data from the user to the serving BS, to be used as input for service execution at the BS, whose output must then be downloaded and routed 
for consumption by the user. Such required communication imposes another restrictive constraint, 
for  routing of the service requests. 

\begin{figure}[t]
  \centering
    \includegraphics[width=0.45\textwidth]{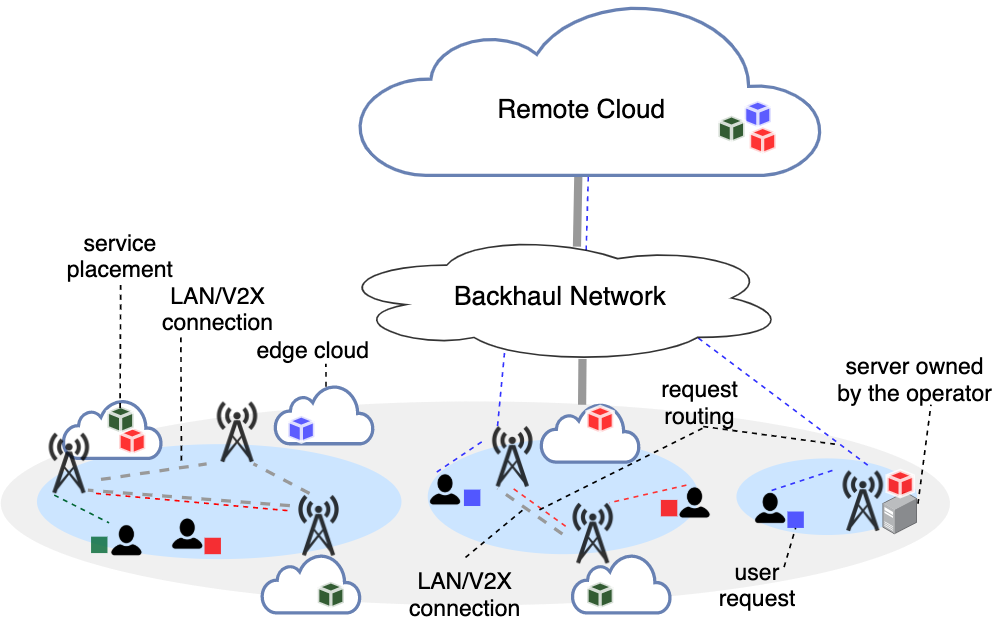} 
      \caption{An example MEC system leased and owned by the MNO with service placement and request routing.}
      \label{fig:network}
\end{figure}

Several existing studies focus on service placement and request routing in cooperative MEC, see e.g.,  \cite{somesula2023cooperative, xu2022near, poularakis2020service, xu2018joint}. These studies primarily address limitations related to edge nodes, including storage, computation, or communication constraints, as mentioned earlier. However, only a few also consider the diverse delay requirements of services \cite{nguyen2021two, wei2021joint, poularakis2020service, farhadi2021service, wen2020joint, somesula2023cooperative}, with most assuming that serving and network delays are constant irrespective of the workload. This assumption is not realistic, particularly for the increasingly emerging delay-sensitive services and applications. 

Furthermore, the existing solutions mainly focus on maximizing the performance of MEC in terms of different performance metrics in view of the network users or the cloud, 
such as maximizing the number of served requests (a.k.a., cache hits) \cite{farhadi2021service},  minimizing the average latency \cite{chen2021latency,somesula2023cooperative},  
or minimizing the cloud traffic \cite{poularakis2020service, nguyen2021two},  to mention a few. 
However, 
the economic 
perspective 
of the problem has been hitherto overlooked, which is indeed a critical aspect 
of cooperative MEC system designs, 
especially for  emerging and future networking applications. 
In particular, the SP, utilizing network resources to deliver MEC-enabled services, incurs costs for the used resources from MNOs and the cloud provider, and earns revenues from the served requests.  
Consequently, maximizing the SP's profit emerges as a pivotal 
factor in the efficiency of the MEC system design, an aspect that, 
has remained unexplored in the existing literature.  

Finally and most importantly, the existing works mainly present heuristic solutions, which does not have any theoretical performance guarantees and hence, may not perform well in practice. However, proposing optimal approaches for the problem of joint optimization of service placement and request routing is yet an important missing work, which, to the best of our knowledge, has not been touched.

In this work, we aim at designing a cooperative MEC system 
which maximizes the total profit of the SP 
by efficiently utilizing the limited 
multi-dimensional resources at the BSs 
while concurrently satisfying the heterogeneous services requirements. 
The key contributions of this work 
can be identified as follows.

\begin{enumerate}
\item We study joint optimization of service placement, CPU sizing, and request admission and scheduling in a cooperative MEC  system, from the perspective of the SP, which delivers heterogeneous MEC-enabled delay-sensitive services across these networks. The SP needs to pay for the used resources to the MNOs and the cloud provider, and earns revenue from the served requests. 
\item We formulate the problem of maximizing the SP’s total profit subject to the computation, storage, and communication constraints of each edge node and end-to-end delay requirements of the services. 
To address the needs of delay-sensitive services, 
we incorporate proper queuing models accounting for the randomness in request arrivals, as well as the service and congestion delays associated with each service.
This formulated problem is a challenging mixed-integer non-convex optimization problem, and we prove it to be NP-hard. 
\item To address the non-convexity challenge in solving the problem, we first introduce a design trade-off parameter for different delay requirements of each service, which maintains a degree of freedom for the priority of satisfying different service delay requirements. Then, by exploiting a hidden convexity, we reformulate the delay constraints into an equivalent convex form. 
\item To handle the challenge of the complicating integer variables, using primal decomposition, we decompose the problem into an equivalent form consisting of master and inner sub-problems, 
where the inner sub-problem contains the simple (i.e., the real) variables only, for fixed values of the complicating integer  variables.  We then employ a cutting-plane approach for building up adequate representations of the extremal value of the inner problem as a function of the 
integer variables and the set of values of the 
integer variables for which the inner problem is feasible. 
\item We propose an iterative 
solution strategy based on generalized Benders decomposition (GBD), and prove that the proposed algorithm converges to the optimal solution within a limited number of iterations. 
\item Finally, using extensive simulation results we demonstrate that the proposed scheme significantly improves the delay, cache hit ratio, load on cloud, and SP’s profit performance compared with various existing mechanisms. 

\end{enumerate}

The rest of the paper is organized as follows. Section~\ref{sec:rw} discusses  related works in the literature. 
After discussing related works in \ref{sec:rw}, we describe the considered system model in Section~\ref{sec:model}. Next in Section~\ref{sec:problem}, we formulate the problem of joint service placement and routing in a cooperative mobile edge computing (MEC) system, and discuss its time complexity. Section~\ref{sec:solution} presents the proposed solution strategy, followed by its convergence analysis in Section \ref{sec:convergence_analysis}. Extensive numerical results and discussions are presented in Section~\ref{sec:evaluation}. Finally, Section~\ref{sec:conclusion} concludes the paper.

%% file: part2.tex
\section{Related Work}\label{sec:rw}

{\scriptsize{
\begin{table*}[!h]
\captionsetup{justification=centering, labelsep=newline}
\begin{adjustwidth}{-.5in}{-.5in}
\caption{Comparison of Related Works}
\label{tab:solutions}
\centering
\begin{tabularx}{0.99\textwidth}{@{}l X c c c c c c c c c c c c c c c c c X@{}} \toprule
 \rotatebox{0}{Ref.} & \rotatebox{0}{Objective} & \rotatebox{90}{BSs Coverage Overlap} &\rotatebox{90}{Cooperative} & \rotatebox{90}{Request Admission} & \rotatebox{90}{Service Placement} & \rotatebox{90}{Request Scheduling} & \rotatebox{90}{CPU Sizing}  & \rotatebox{90}{Target Max. Delay$^1$} & \rotatebox{90}{Communication Delay} & \rotatebox{90}{Serving Delay} & \rotatebox{90}{Queuing Model} & \rotatebox{90}{Communication Capacity} & \rotatebox{90}{Node Storage Capacity} & \rotatebox{90}{Node Computation Capacity} & \rotatebox{90}{Migration Cost} & 
 \makecell{Solution Type,\\Method}
 \\
 \bottomrule
 \textbf{This work} 
 & Max SP's profit & \no & \yes & \yes & \yes & \yes & \yes & \yes & \yes & \yes & \yes & \yes  & \yes & \yes  & \no & 
 \makecell{Optimal, \\ GBD 
 }\\
 \bottomrule
 \cite{somesula2023cooperative} & Max time utility &\yes &\yes &\no &\yes &\yes &\no &\no & \no &\no & \no &\yes &\yes & \yes & \no &\makecell{Heuristic,\\Randomized Rounding \\ Circular Convex Set}\\ 
 \hline
 \cite{poularakis2020service} & Max num. of edge-served requests& \yes & \no & \no & \yes & \yes &\no  & \no & \no & \no & \no & \yes  & \yes & \yes &\yes& \makecell{Approximate, \\ Randomized rounding}\\
 \hline
  \cite{farhadi2021service} & Max edge served requests& \no & \yes & \no & \yes & \yes & \no & \no & \no & \no & \no & \yes & \yes & \yes &\yes& \makecell{Approximate,\\ Submodularity$^2$} \\
  \hline
 \cite{xu2018joint}& Min long-term average delay 
& \yes & \no & \no & \yes & \no$^{3a}$ & \no & \yes$^{3b}$ & \yes & \yes & \yes & \no & \yes & \yes &\no & \makecell{Heuristic, \\Lyapunov optimization \\ Gibbs sampling}    \\
\hline
 \cite{ma2020cooperative} & Min delay and cloud traffic & \no & \yes & \no & \yes & \yes & \no & \no & \yes & \yes & \yes & \yes & \yes & \yes &\no & \makecell{Heuristic, \\ Gibbs sampling and water filling} \\
 \hline
\cite{chen2021latency} & Min average latency& \no & \no & \no & \yes & \yes &\yes  & \no & \yes & \yes & \yes & \yes & \yes & \yes &\yes& \makecell{Approximate,\\Convex approximation}\\
\hline
\cite{nguyen2021two} & Min total cost of SP for one service&\no&\no&\no&\yes&\yes&\yes&\yes$^4$&\yes&\no&\no&\no&\yes&\yes&\yes&\makecell{Optimal,\\Column \& Constraint Generation}\\
\hline
\cite{wen2020joint} & Min avg. energy consumption &\no&\no&\no&\yes&\yes$^{5a}$&\no&\yes$^{5b}$&\yes&\yes&\no&\yes&\yes&\yes&\yes&\makecell{Heuristic,\\ADMM}\\
\bottomrule
\end{tabularx}
\end{adjustwidth}
\footnotesize{
$^1$ Maximum end-to-end delay which traffic packets of a service can tolerate.\\ 
$^2$ For the case when there is no contention of computational resources, \ie, when there is abundant computational resources or or when each edge node can store up to one service. 
$^{3a}$ Traffic is equally divided between the edge nodes which have cached the service, $^{3b}$ Long-term average delay constraint. 
$^4$ Average delay constraint or delay budget. 
$^{5a}$ One edge node, $^{5b}$ Symmetric deadline. }
\end{table*}
}}


There exists extensive literature on resource management for multi-tier architectures (consisting of end user devices, edge nodes or aggregation points, and the cloud), all seeking the overall goal of taking full advantage of the distributed storage and the computation and communication capabilities of the multi-tier architectures, while assuring QoS requirements of services and resource constraints.

A first body of works concentrate on optimal content placement in cache hierarchies, \eg, see  \cite{borst2010distributed, ahmadi2020cache}. In particular, \citet{borst2010distributed} consider cooperative caching scheme under storage constraints, and \citet{ahmadi2020cache} investigate optimal memory-bandwidth trade-off given unit bandwidth and storage prices. However, these works consider only storage and communication constraints and do not take computation-intensive services into account.  
Relevant to our work are also studies that consider computation/task offloading from resource-constrained end-user devices to the edge nodes (a.k.a. 
cloudlets)  or to the cloud in order to improve battery efficiency and reduce latency \cite{ndikumana2019joint, yang2019cloudlet, eshraghi2019joint}. Typically, these works artificially assume that the heterogeneous execution platforms and softwares required by the tasks are readily available at the edge nodes too (similar to the cloud). As a result, the task can either be executed locally or remotely (\ie, on the edge or on the cloud) depending on the objective and constraints.
\citet{ndikumana2019joint} consider just big data tasks, and propose an $\epsilon$-optimal cooperative input data caching and computation offloading mechanism to maximize the outsourced bandwidth saving and minimize the latency inside each cluster of the edge nodes under storage, computation, and serving delay constraints.       

The closest  works to ours are those considering joint optimization of service placement (or caching) and request scheduling between edge nodes and the cloud, e.g., see \cite{ poularakis2020service, farhadi2021service, xu2018joint, ma2020cooperative, somesula2023cooperative, wen2020joint, nguyen2021two, wei2021joint, zhang2019dmra, pasteris2019service}. Some of these works consider only one specific type of service, \eg, see \cite{nguyen2021two} and \cite{wei2021joint}. The works in  \cite{zhang2019dmra} and \cite{pasteris2019service} further assume abundant storage and computational resources, respectively, which are  not realistic. The works in \cite{ poularakis2020service, farhadi2021service, xu2018joint, ma2020cooperative, somesula2023cooperative, wen2020joint} consider multiple 
services under storage, computation, and communication constraints. In more details, \cite{poularakis2020service} and \cite{farhadi2021service} try to maximize the number of served requests by edge servers. 
However, they ignore the latency, and hence, their proposed schemes are not appropriate for delay-sensitive applications. \citet{xu2018joint} propose a Lyapunov-based dynamic solution to minimize the long-term average delay under long-term average delay and energy constraints. Likewise, \citet{ma2020cooperative} proposes a cooperative solution based on Gibbs Sampling \cite{lynch2007introduction} to minimize the combination of delay and outsourcing traffic to the cloud. 
\citet{chen2021latency} consider a three-layer architecture where non-admitted requests can be executed locally, and targets minimizing average delay. \citet{wen2020joint} consider a similar network structure, but with one edge node, where all the tasks arrive at the same time and have to be completed by a certain time. Compared to the  previous work, they also consider software fetching and multicasting for local execution, and devise a heuristic algorithm based on alternating direction method of multipliers (ADMM) \cite{boyd2011distributed} to minimizes the average energy consumption.  

It should be noted that most of the works mentioned above 
assumes constant serving and network delays, regardless of workload \cite{nguyen2021two, wei2021joint, poularakis2020service, farhadi2021service, wen2020joint, somesula2023cooperative}. 
However, practical scenarios involve stochastic request arrivals, leading to delays influenced by queue dynamics.
\citet{xu2018joint} model each edge node as one M/G/1 queue which is shared between different service types. 
Despite its higher utilization and lower average latency, this mechanism lacks service isolation (thus, all the services experience the same waiting delay), making it unsuitable for services with strict delay requirements. 
\citet{ma2020cooperative} consider different M/M/1 queues with a specific service rate for each service on an edge node.They also model the congestion on the backhaul link with an M/M/1 queue. However, the amount of computational resources allocated to each service (named as the CPU size in this paper) is not optimized.  
\citet{chen2021latency} allocate a queue of type M/G/1 per service and per mobile device on the edge. 
Additionally, to characterize backhaul congestion, they consider an M/G/1 queue per service and per edge node. Moreover, taking the warm-up time into account, they assume service times follow a shifted exponential distribution. 
Note that all the aforementioned works assume negligible serving latency at the remote cloud due to the  abundant resources. 
However, 
taking the economic aspect of the problem 
 into consideration, this assumption is unrealistic as the SP is required to pay for used resources. 
 
So far, support for delay-sensitive services has attracted less attention in the literature. It should be noted that even minimizing the average delay \cite{ma2020cooperative, chen2021latency} or considering some average delay constraints \cite{xu2018joint, nguyen2021two} can not guarantee individual stringent delay requirements of different services.  
Moreover, although \citet{wen2020joint} concentrate on delay-sensitive services, 
they overlook the impact of random request arrivals. Additionally, their approach assumes identical delay requirements for all services. Yet, 
none of the previous works has characterized computational delay in the problem formulation.

Finally, regarding the solution strategy, given the inherent complexity of the optimization problems within the realm of MEC,  existing works predominantly propose heuristic solutions, which may not demonstrate high efficiency in practical applications. There are limited attempts to devise near-optimal solutions in this context. Notably, works in \cite{zhao2018benders, yang2019cloudlet} utilize linear Benders decomposition to address some formulated mixed-integer linear programming (MILP) problems in this domain. However, their theoretical performance guarantees are limited to linear problems, while the comprehensive MEC optimization problem considered in this work is a challenging mixed-integer non-linear (even non-convex) problem. Another 
strand of works employs randomized rounding approximation techniques for service placement and request routing \cite{poularakis2020service, somesula2023cooperative, xu2022near}. Nevertheless, their theoretical performance guarantees pertain to linear problems only, and their approximation solutions still exhibit a significant optimality gap. Consequently, the quest for optimal or close-to-optimal solutions for  optimization of MEC systems design remains unfulfilled, which stands as one of the pivotal  motivations driving the focus of this paper.
Table I provides a comprehensive comparison of the works most closely related to ours.

%% file: part3.tex


\begin{table}[h!] 
\caption{Frequently Used Notations.}
\centering
\begin{tabular}{l p{0.7 \columnwidth} }
 Symbol & Description  \\ \hline
 \multicolumn{2}{c}{Parameters} \\\hline
 $e \in \mathbb E$ & BS index\\ 
 $n \in \mathbb E \cup \{c\}$ & server  index ($ c $ denotes the cloud) \\ 
 $l_{en} \in \{0, 1\}$&indicates whether BS $e$ and server $n$ are connected \\
 $s \in \mathbb S$ & service index \\
 $C_e$ & storage limit at $\mathrm{BS}_e$ \\
 $B_e$ & communication limit at $\mathrm{BS}_e$\\
 $B^l_{en}$ & uplink/downlink communication limit at link $en$\\
 {{$M_e$}} & computation limit at $\mathrm{BS}_e$ \\
 $ d^{\text{prp}}_{en}$ & propagation delay of link $en$\\
 $\lambda_{es}$  & demand for service $s$ arrived at BS $e$ \\
 $\sigma_{s}$& size per replica of service $s$ \\
 $\beta_{s}$& size of input/output data per request of $s$\\
 $\nu_{s}$& average  computation complexity per request of $s$ \\
 $D^{\text{max}}_{s}$& target delay of service $s$ \\
 $p^{\textit{cpu}}_{n}$& unit cost of computation at server $n$\\
  $p^{\textit{str}}_{n}$& unit cost of storage at server $n$\\
 $p^{\textit{trn}}_{en}$ & unit traffic cost at link $en$ \\\hline
 \multicolumn{2}{c}{Decision Variables} \\\hline
 $x_{ns} \in \{0, 1\}$ & placement indicator variable of service $s$ at server $n$\\
 $y_{sen} \in [0,\,1]$ & 
 the fraction 
 of  requests for service $s$ 
 at server $e$ that is delegated to server $n$\\
 $u_{ns} \in \mathbb{R}^+ $& CPU size allocated to service $s$ at node $n$ \\\hline
  \multicolumn{2}{c}{Auxiliary Variables} \\\hline
$D^{\text{srv}}_{ns} \in \mathbb{R}^+$& average service delay of requests for service $s$ at server $n$ \\
$D^{\text{cng}}_{en} \in \mathbb{R}^+ $ & average congestion delay of link $en$\\
$\Lambda_{ns} \in \mathbb R^+$ & total requests for service $s$ delivered to server $n$\\
 $\Lambda^l_{en} \in \mathbb R^+$ & uplink/downlink traffic passing over link $en$\\ 
 $z_{sen} \in \{0, 1\}$& an indicator representing whether any request for service $ s $ originating at BS $ e $ is transferred to server $ n $  
 (i.e, $ z_{sen} = \mathds{1} [y_{sen}>0] $).
 \\\hline
\end{tabular}
\label{table1}
\end{table}

\section{System Model}\label{sec:model}

We study service placement, CPU sizing, and request admission and scheduling in a cooperative MEC system provided by one or possibly multiple mobile networks. This problem is modeled taking the perspective of a service provider (SP) which delivers MEC-enabled delay-sensitive services across these networks.
MNOs deliver edge-located computation and storage resources. These facilities are pre-installed and collocated with the operators' network, \eg, at BSs and aggregation sites  \cite{hu2015mobile}. 
SP should pay for used resources across the MEC network and the cloud to the MNOs and the cloud provider, respectively.  

We consider a 2-tier network structure consisting of a remote cloud $c$ with abundant resources and a cooperative MEC system with a set of BSs, denoted by $\mathbb E=\{1,\cdots, E\}$,  distributed at fixed locations in a large service region, as illustrated in \figurename~\ref{fig:network}. Let $n \in \{\mathbb E, c\}$ denote the index of a server 
(either a BS or the cloud). 
The BSs are connected to their neighbouring BSs via low latency, high bandwidth LAN connection, and to the remote cloud via backhaul links. The existence of such a direct link between each $\mathrm{BS}_e$ and server $n$ is denoted by an indicator variable $l_{en}$, and  $B^l_{en}$ indicates the bandwidth capacity of the link. Here, all the BSs are assumed to be connected to the central remote cloud, and hence, we have $l_{ec}=1$ for all $e \in \mathbb E$.


Each $\mathrm{BS}_e$ is equipped with limited storage (hard disk), computation, and communication capacities indicated by $C_e$, $M_e$, and $B_e$, respectively. SP can lease storage and computation resources at server $n \in \{\mathbb E, c\}$ at unit costs $p^{\text{str}}_{n}$ and $p^{\text{cpu}}_n$, respectively. 
MNOs also charge SP at $p^{\text{trn}}_{en}$ per unit of bandwidth consumption between $\mathrm{BS}_e$ and server $n$. 
The computation cost at the cloud is typically cheaper than the  BSs. Furthermore, the aforementioned costs can also be different across different BSs depending on the location, owner MNO and its reputation, reliability, etc.    



The SP provides a library of $\mathbb S=\{1, \cdots, S\}$ latency-sensitive services, \eg, online gaming, autonomous driving, video streaming, which are heterogeneous in terms of storage, computation, and communication requirements \cite{huawei}. For clarity, we focus on the downlink traffic in this paper, but the proposed method can be extended to take into account uplink as well without loss of generality. All BSs might receive requests for all the services and demand is assumed independent of service placements. We assume requests for each service $s$ arriving at $\mathrm{BS}_e$  
follow a Poisson process with average rate $\lambda_{es}$, which is a common assumption in existing works \cite{xu2018joint, ma2020cooperative}. Each service $s$ requires $\sigma_s$ storage units per replica to store its codebase and  required databases and its output data size per request is $\beta_s$. 
Furthermore, services can have different processing requirements depending on their computational complexity. The computational capability (in CPU cycles) required to process one request of service $s$ is assumed to be an exponential random variable with average $\nu_s$. 
Finally, services come from different applications with divergent delay requirements. For example, an online gaming service can tolerate less delay than a video streaming service.
We consider a target delay of $D^{\text{max}}_s$ for each service $s$, 
which defines the maximum tolerable end-to-end delay for requests of service $s$. 
Table~\ref{table1} shows the list of our frequently  used  notations.

\section{Problem formulation}\label{sec:problem}

Given the knowledge of the service request arrivals and MEC topology, the SP should decide about \begin{enumerate*}[label=\itshape\roman*\upshape)]
    \item service placement, \ie, which servers should run each type of service, \item workload admission and scheduling, \ie, which fraction of the received workload at each BS should be admitted and how it should be distributed across the neighboring BSs and the cloud, and \item CPU sizing, \ie, how much computation capacity should be allocated to each virtual machine for running its associated 
    service.
\end{enumerate*}
Following the previous works (\eg, see \cite{poularakis2020service}), we assume that the aforementioned  decisions are taken at the beginning of a long-time window during which the request rates are assumed to be fixed and predictable. In the following, we formalize these decisions, and introduce the objective function  and the system constraints in the problem of joint service placement, request admission and scheduling, and CPU sizing in a cooperative MEC system. 

\subsection{Decision  Variables and Constraints}

\textbf{Service Placement.}
Service placement is defined by the following matrix  
\begin{equation*}\label{eq:variable:x}
{\mathbf{X}} = [x_{ns} \in \{0, 1\}: n \in \{\mathbb E, c\}, s \in \mathbb S ],
\end{equation*}
where $x_{ns} = 1$ if service $s$ is stored in server $n$ and $x_{ns} = 0$ otherwise. Note that since the cloud has all the services, we always have $ x_{cs} = 1, ~\forall s $. Clearly, the total size of the  services placed on each $\mathrm{BS}_e$ should not exceed its storage capacity, \ie, 
\begin{equation}\label{eq:capacity:storage}
\sum_{s \in \mathbb{S}} \sigma_s x_{es} \leq C_e, \forall e.
\end{equation}

\textbf{Request Admission and Scheduling.}
For each service and each BS, the SP needs to decide which fraction of requests to admit, and how to schedule the admitted ones across servers, depending on the target delay of the service, its required resources, and leasing costs. 
Request scheduling allows each BS to delegate any fraction of arrived requests for a service to the cloud or any of its neighbouring BSs that have that service,
allowing flexibly and efficiently utilizing the computational and storage capacities of the entire network.

In order to formalize the request admission and scheduling, we define a service delegation matrix as 
\begin{equation*}\label{eq:variable:y}
{\mathbf{Y}} = [y_{sen} \in [0, 1]: s \in \mathbb S, e \in \mathbb E, n \in \{\mathbb E, c\}],
\end{equation*}
where $y_{sen}$ indicates the fraction of requests for service $s$ arrived at $\mathrm{BS}_e$ that are admitted and delegated to server $n$ for being served. 
According to the above definition, clearly, at each $\mathrm{BS}_e$ and for each service $s$, we have the following normalization constraint over different associated delegation control variables, which indicates the sum of admitted requests cannot exceed $1$:
\begin{equation}\label{eq:sum1}
\sum_{n \in \{\mathbb{E}, c\}} y_{sen} \leq 1, \forall s, e.
\end{equation}
Moreover, the communication capacity constraint of each $\mathrm{BS}_e$ should also be considered in the request admission, as follows
\begin{equation}\label{eq:capacity:bw}
\sum_{s \in \mathcal S}  \beta_s \sum_{n \in \{\mathbb E, c\}}  \lambda_{es} y_{sen} \leq B_e, \forall e.
\end{equation}
Finally, $\mathrm{BS}_e$ can delegate  requests for service $s$ to a server $n$ only if the service has been stored in that server and the server is  a neighbouring BS or the cloud, \ie,  
\begin{equation}\label{eq:routing}
y_{sen} \leq x_{ns} l_{en}, \forall s, e, n.
\end{equation}

\textbf{CPU Sizing.} Let matrix $\mathbf{U}$ represent the CPU sizing decision variable as 
\begin{equation*}\label{eq:variable:mu}
    {\mathbf{U}} = \{u_{ns} \in \mathbb R_+: n \in \{\mathbb E, c\}, s \in \mathbb S\},
\end{equation*}
where $u_{ns}$ represents the computation capacity (\eg, in GHz) reserved on server $n$ 
for running service $s$. 
At each BS, the computation capacity allocations should preserve the BS's computation capacity constraint, \ie, a feasible CPU sizing should satisfy the following constraint
\begin{align}
    \sum_{s \in \mathbb S} u_{es} &\leq M_e, \forall e.\label{eq:capacity:mu}
\end{align}
Note that the SP leverages virtualization technologies supported by the MEC system (\eg, virtual machine (VM), container, etc)  to deploy each service in isolation on the general-purpose server \cite{mecAndSlice}. This separation will ensure the safe operation, security, and privacy of running services, especially in terms of services with rigid QoS requirements or with access to customer information.

\subsection{Delay Constraints} 

When a request for service $s \in \mathbb S$ is being served by server $n \in \{\mathbb E, c\}$, the end-to-end latency that it experiences should be less than the target delay of service $s$, \ie, $D^{\text{max}}_s$. 

The end-to-end latency that  a service request experiences includes three parts: 1) \textit{the propagation delay:} the round-trip propagation time of the direct link between the original BS at which the request has arrived and the destination server which is going to serve the request, 2) \textit{the congestion delay:} the total time it takes for the request to be successfully transmitted to the destination server on the interconnecting link between the two servers, and 3) \textit{the service delay:} the processing delay at the destination server.  
We note that 
similar to the previous works (e.g., see \cite{yang2019cloudlet, nguyen2021two, ma2020cooperative}), the round-trip delay between a  user and its covering BS is assumed to be negligible. 

Let $D^{\text{srv}}_{ns}$ and $D^{\text{cng}}_{sen}$ denote the service delay of service $s$ at server $n$, and its congestion delay on link $en$, respectively, and $d^{\text{prp}}_{en}$ be the propagation delays of link $en$. 
Therefore, the target delay constraint can be represented as, 
\begin{equation}\label{eq:delay}
     (D^{\text{srv}}_{ns} + d^{\text{prp}}_{en} + D^{\text{cng}}_{sen}) \mathds{1}_{[y_{sen}>0]}  \leq D^{\text{max}}_s, \,\forall s, e , n, 
\end{equation}
where $\mathds{1}_{[y_{sen}>0]}$ is an indicator function which equals to $1$ only when $y_{sen}>0$.

In the rest of this subsection, we derive the delay terms in the left hand side of the delay constraint in \eqref{eq:delay}. It should be noted that 
the propagation delay of all the network links are always fixed and can be assumed to be known, while the service and congestion delays highly depend on the admisison and scheduling policy. Therefore, it remains to derive the expressions of $D^{\text{srv}}_{ns} $ and $D^{\text{cng}}_{sen}$ as functions of the problem parameters and control variables. For this purpose, first note that as the arrival of requests for each service type $s$ at each $\mathrm{BS}_e $ is assumed to follow an independent Poisson process, with parameter denoted by $\lambda_{es}$,
the forwarded requests of service $s$ from any $\mathrm{BS}_e $ to a server $n$ is also a Poisson process with the new parameter $\lambda_{es}y_{sen}$, and hence,  the total requests for service $s$ served at server $n$ can also be modeled as a Poisson process with a rate of $\Lambda_{ns} = \sum_{e \in \mathbb E}\lambda_{es}y_{sen} $.
Similarly, the total downlink traffic over link $en$ follows a Poisson process with rate $\Lambda^l_{en} = \sum_{s \in \mathbb S}\beta_s \lambda_{es}y_{sen}$.


From the definitions of $\nu_s$ and $ u_{ns}$, it follows that the processing time of each request for service $s$ at server $n$  is an exponentially-distributed random variable with the rate of $u_{ns}/\nu_s$. Therefore, the service delay of the requests for $s$ at server $n \in \{\mathbb E,c\}$ can be estimated by the average response time of an M/M/$1$ queue \cite{ross2014introduction} as $ D^{\text{srv}}_{ns} =   \frac{1}{u_{ns}/\nu_s - \Lambda_{ns}}, \forall n, s $. 
Likewise, any link $en$ can be modeled as an M/M/1 queue with arrival rate of $\Lambda^l_{en}$ and service rate of  $B^l_{en}$. 
Consequently, the congestion delay at link $en$ is derived as $ D^{\text{cng}}_{sen} =   \frac{\beta_s}{B^{l}_{en}-  \Lambda^l_{en}}, \forall s, e, n, n\neq e $, 
and $D^{\text{cng}}_{see} =0, \forall s, e$.
\footnote{
Please note that, assuming symmetric uplink/downlink bandwidth for network links and symmetric input/output size for services, the round-trip congestion delay at link $en$ can be considered as twice the congestion delay of the uplink or downlink, as demonstrated above. Nevertheless, the computation of the congestion delay in the asymmetric case is also straightforward.}

Finally, note that to ensure the stability of the aforementioned queues, we need the following two constraints: 
\begin{align}
    \Lambda_{ns} &\leq u_{ns}/\nu_s, \forall n,s, \label{eq:stability:node}\\
     \Lambda^l_{en} &\leq B^l_{en}, \forall e, n, n\neq e. \label{eq:stability:link}
\end{align}

\subsection{The Target Optimization Problem}
We aim at maximizing the SP's total profit gained from serving the services while minimizing the total cost of leased storage, computation, and communication resources. The services may belong to usecases with different latency, priority, and resource requirements. The SP will account for these differences while monetizing services.

The SP's total serving revenue per unit of time can be expressed as
\begin{equation}\label{eq:revenue}
    \mathcal R_{\text{srv}} ({\bf{Y}}) = \sum_{s \in \mathbb S} \sum_{n \in \{\mathbb{E}, c\}}  R^{\textit{srv}}_{ns} \left(  \Lambda_{ns} \right), 
\end{equation}
where $R^{\textit{srv}}_{ns}$ is a twice continuously differentiable, 
concave and increasing function of the  total  requests  for service $s$ served at server $n$. Note that the above assumptions on the revenue function can capture a lot of interesting cases such as $\alpha$-fairness and proportional fairness \cite{omidvar2018optimal}.

Moreover, the total cost includes the costs of leased storage, computation, and communication resources, which are described per unit of time, respectively, as:
\begin{align}
\mathcal P_{\text{str}}({\mathbf{X}}) &= \sum_{s \in \mathbb{S}} \sigma_s \sum_{n \in \{\mathbb{E}, c\}}  p^{\text{str}}_n  x_{ns},\label{eq:costs:str}\\
\mathcal P_{\textit{cpu}}({\mathbf{U}}) &= \sum_{s \in \mathbb{S}} \sum_{n  \in \{\mathbb E, c\}} p^{\textit{cpu}}_n   u_{ns},\label{eq:costs:cpu}\\
\mathcal P_{\text{trn}}({\mathbf{Y}}) &= \sum_{e \in \mathbb E}\sum_{n \in \{\mathbb E, c\}}  \Lambda^l_{en}p^{\textit{trn}}_{en}. \label{eq:costs:trn}
\end{align}
Accordingly, the SP's total profit  
can be derived as 
\begin{align}
  \hspace{-4 pt}  \mathcal{U}_{\text{total}} \left( \mathbf{X},\mathbf{Y},\mathbf{U} \right)
    =
    \mathcal R_{\text{srv}}({\mathbf{Y}})  -\mathcal P_{\text{cpu}}({\mathbf{U}}) -\mathcal P_{\text{trn}}({\mathbf{Y}}) - \mathcal P_{\text{str}}({\mathbf{X}}).
\end{align}
Finally, the problem of joint service placement, request admission and scheduling, and CPU sizing for the SP's profit maximization can be formulated as
\begin{align}\label{eq:orig. opt.}
\mathcal{P}_{org}: \quad    &\max_{\mathbf{X,Y,U}} 
\mathcal{U}_{\text{total}} \left( \mathbf{X}, \mathbf{Y},\mathbf{U} \right)
\notag\\
    &\text{subject to}\, \eqref{eq:capacity:storage}-\eqref{eq:delay}, \eqref{eq:stability:node},\eqref{eq:stability:link}.
\end{align}

\subsection{Problem Complexity Analysis}

Problem $ \mathcal{P}_{org} $ includes both integer variables $\mathbf{X}$ and continuous variables ($\mathbf{Y}, \mathbf{U}$), and hence, it is a mixed-integer non-linear programming (MINLP) problem. 
In the following, we will describe a simple instance of $\mathcal{P}_{org}$ and prove that it remains NP-hard even in this simplified case. This will, in turn, demonstrate that $ \mathcal{P}_{org}$ is NP-hard in the general form \eqref{eq:orig. opt.}, and hence, cannot be solved in polynomial time.   

Consider a simple case where we do not take account of cost terms   \eqref{eq:costs:str}-\eqref{eq:costs:trn} in the objective function, \ie, all unit leasing costs $p^{str}_n$, $p^{\textit{cpu}}_{n}$, and $p^{\textit{trn}}_{en}$ are assumed to be $0$, and we consider linear and homogeneous revenue functions for each $s$, \ie, $R^{\textit{srv}}_{ns}=R^{\textit{srv}}_{s}=r^{\textit{srv}}_{s} \Lambda_{ns}, \forall n, s$. Setting capacities $B_e$, $M_e$, and $B^l_{en}$, $\forall e, n$, large enough eliminates constraints \eqref{eq:capacity:bw}, \eqref{eq:capacity:mu}, and  \eqref{eq:stability:link}, respectively, rules out decision variable matrix ${\mathbf{U}}$ and associated constraint \eqref{eq:stability:node}, and decrease delay terms $D^{\text{srv}}_{ns}$ and $D^{\text{cng}}_{sen}$ to zero in \eqref{eq:delay}. 
Moreover, we assume all BSs are connected by direct links with zero propagation delay, \ie, $l_{ee'} = 1$ and $d^{\text{prp}}_{ee'}=0, \forall e, e' \in \mathbb E$, and $d^{\text{prp}}_{ec} > \max_{s \in \mathbb S} {D^{\text{max}}_s}, \forall e$. 
Based on aforementioned assumptions,  \eqref{eq:delay} and \eqref{eq:routing} are satisfied for any $n \in \mathbb E$ and \eqref{eq:routing} is violated for $n=c$.
Therefore, in the optimal solution, no service will be placed on the cloud, and there is at most one instance for each service placed on all BSs, \ie,  $\sum_{e \in \mathbb E}{x_{es}} \leq 1, \forall s$. If $x_{es}=1$, all request for that service will be delegated to $e$, \ie, $y_{se'e}=1, \forall e' \in \mathbb E$, and $y_{se'e}=0$ otherwise. Thus, we have $y_{se'e}=x_{es}, \forall s, e, e'$, and   
$\mathcal{P}_{org}$ can be rewritten as
\begin{align}
\mathcal{P}_{sp}: \quad&\max_{{\mathbf{X}} } \sum_{s \in \mathbb S} \sum_{e \in \mathbb{E}}  r^{\textit{srv}}_{s} x_{es} \sum_{e' \in \mathbb{E}} \lambda_{e's},\notag\\
\text{subject to}\, &\eqref{eq:capacity:storage}.\notag
\end{align}

\begin{theorem}[Problem Complexity]
Problem $\mathcal{P}_{\text{sp}}$ (which is a simple instance of $\mathcal{P}_{org}$) is NP-hard. 
\end{theorem}

\begin{proof}
We prove $\mathcal{P}_{\text{sp}}$ is NP-hard by reduction from the NP-hard 0-1 multiple knapsack problem (MKP)  defined as \cite{martello1990knapsack}: Given a set of $n$ items, where the weight of the $i$th item is $w_i$ and its profit is $v_i$, and a set of $m$ knapsacks ($m < n$), each having a maximum capacity $\hat{c}_j$, select $m$ disjoint subsets of items $\mathcal S_j$ such that the total profit of selected items, \ie, $\sum_{j=1}^{m} \sum_{i \in \mathcal S_j}v_i$, is maximum, and each subset $S_j$ can be assigned to a different knapsack $j$ whose capacity is not less than the weight of all items in the subset, \ie, $\sum_{i \in \mathcal S_j}w_i \leq c_j$. 

For an instance of MKP($[v_i]$, $[w_i]$, $[\hat c_j]$), we construct an instance for the special case of $\mathcal{P}_{\text{sp}}$ as follows: we associate each item $i$ to service  $s_i$ with $v_i=\sum_{e \in \mathbb{E}} \lambda_{es_i}$ and $w_i=\sigma_{s_i}$. For each knapsack $j$, we also consider a $\mathrm{BS}_j$ with $C_j=\hat c_j$.  
Let ${\mathbf{X}}$ be the optimal solution to $\mathcal{P}_{\text{sp}}$ and $\mathcal S_j=\{i \in \{1, \cdots, n\}: x_{js_i}=1\}$ be the set of all placed services on $\mathrm{BS}_j$. Since optimal solution will place each service on at most one BS and will preserve constraint \eqref{eq:capacity:storage}, it is easy to see that ${\mathbf{X}}$ is the optimal solution to $\mathcal{P}_{\text{sp}}$ if and only if $S_j$ ($j=1,\cdots,m$) is the optimal solution of MKP. 
\end{proof}

%% file: part4_v2.tex
\section{Solution Strategy}\label{sec:solution}

As mentioned in the previous section, the formulated problem is an MINLP, where the coupling between the integer 
and continuous variables complicates the problem. 
Moreover, due to 
the constraints of the form \eqref{eq:delay}, 
even relaxing the integer variables cannot help to address this problem in polynomial time, as the relaxed version of the problem still remains non-convex.
Consequently, the conventional mixed-integer linear programming (MILP) or convex programming approaches cannot handle this problem.   

To overcome the aforementioned challenges efficiently, we introduce a new design trade-off parameter for different delay requirements of each service, providing flexibility in prioritizing different service delay requirements. 
By exploiting a hidden convexity, we reformulate the delay constraints into an equivalent form. Next, to handle the complicating integer variables, we use primal decomposition, decomposing the problem into an equivalent master and inner problems, over   complicating and simple variables, respectively. We employ a cutting-plane approach to build adequate representations of the extremal value of the inner problem as a function of the complicating integer variables and the set of values of the integer variables for which the inner problem is feasible.

Finally, we propose a solution strategy based on generalized Benders decomposition (GBD) and prove its convergence to the optimal solution within a limited number of iterations. Before that, we will briefly review the GBD approach in the next subsection.

\subsection{Problem Reformulation and Decomposition}\label{sec:solution:reformulate}

First note that in the formulated problem  $\mathcal P_{org}$ involves integer variables  $\mathbf{X}$, which complicates the problem. Furthermore, since constraints \eqref{eq:delay} are non-convex, even with relaxing the integer variables (a  common 
step 
in the randomized rounding methods of previous works \cite{poularakis2020service, somesula2023cooperative}) or fixing the complicating variables, the problem still remains MINLP, which remains challenging. 
To tackle the aforementioned challenges, in the following, we first introduce a new design parameter 
that maintains a trade-off between different delay requirements of each service and hence, provides flexibility in prioritizing different service delay requirements. For a design parameter $\alpha \in [0, 1]$ , let $(1-\alpha) D^{\text{max}}_s$ and $\alpha D^{\text{max}}_s$  denote the maximum tolerable service  delay, and congestion plus propagation delays, respectively. By incorporating this design parameter in the delay constraints \eqref{eq:delay}, they will be transformed  into the following set of constraints: 
\begin{subequations}
\begin{align}
       D^{\text{srv}}_{ns} \mathds{1}_{[y_{sen} > 0]} & \leq  D^{\text{max}}_s, \, \forall s, e, n=e,\label{eq: delay const served at its original BS}\\
       D^{\text{srv}}_{ns}  \mathds{1}_{[y_{sen} > 0]} &\leq \left( 1- \alpha \right) D^{\text{max}}_s, \,\forall s, n \neq e,
\end{align}
\end{subequations}
\begin{align}
(D^{\text{cng}}_{sen} + d^{\text{prp}}_{en}) \mathds{1}_{[y_{sen} > 0]} &\leq \alpha D^{\text{max}}_s, \forall s, e, n \neq e, 
\end{align}
where constraint \eqref{eq: delay const served at its original BS} 
is the delay constraint for the case when the request is served at its original BS, i.e., the request encounters only a  service delay at its serving BS, but  
no network delay. 

It is worth noting that summing up the above delay constraints results in the previous constraint \eqref{eq:delay}. However, the above form of delay constraints provides a  more flexible design than that of our originally formulated delay constraint in \eqref{eq:delay}. In fact,  parameter  $\alpha$ maintains a design trade-off 
between the priority of satisfying 
different service requirements. 
For example, when $\alpha$ is set to a large value,  the SP allows to serve a request even if it has to go through a relatively saturated link (i.e., with high network delay) which in turn tightens the delay bound for service delay and hence,  increases the required computational resources. 
%
%
%

Note that the above set of delay constraints are still non-convex. To tackle this challenge, we first define a new set of auxiliary variables $ \mathbf{Z} = [ z_{sen} \in \{ 0,1\}:  s \in \mathbb{S} , e \in \mathbb{E}, n \in \left\{ \mathbb{E}, c \right\} ] $, in which  $z_{sen}$ indicates 
whether any request for service $ s $ originating at BS $ e $ is transferred to server $ n $  
 (i.e, $ z_{sen} = \mathds{1} [y_{sen}>0] $). 
 Then, we can reformulate the above non-convex delay constraints into an equivalent linear form, as follows:
\begin{subequations}
\begin{align}
    z_{sen}&
    \leq D^{\text{max}}_s (u_{ns}/\nu_s - \Lambda_{ns}), \,\forall s, e, n=e, \label{eq:const:delay2_a} \\
    z_{sen}&
    \leq (1-\alpha) D^{\text{max}}_s (u_{ns}/\nu_s - \Lambda_{ns}), \,\forall s, e, n\neq e, \label{eq:const:delay2_b}
\end{align}
\end{subequations}
\begin{align}\label{eq:const:delay2_c}
    z_{sen}
    \leq \max \{\alpha D^{\text{max}}_s - d^{\text{prp}}_{en}, 0\} (B^l_{en} - \Lambda^l_{en}),\, \forall s, e, n \neq e,
\end{align}
\begin{equation}\label{eq:const:delay2_d}
    y_{sen} \leq z_{sen}, ~ \forall s,e,n, 
\end{equation}
where inequality \eqref{eq:const:delay2_c}  
is to ensure that $y_{sen}=0$ when $\alpha D^{\text{max}}_s < d^{\text{prp}}_{en}$. Note that this happens when the parameter $\alpha$ is set to a very small value or the associated service is so delay-sensitive that no network delay is tolerable for serving the requests, and hence, no service request can be delegated from its original BS to another server.  
Moreover,    inequality \eqref{eq:const:delay2_d} is to  guarantee that no portion of the requests for service $s$ arrived at $\mathrm{BS}_e$ can be served by server $n$ if  it is not transferred to that server. 

Finally, we note that having constraint  \eqref{eq:const:delay2_d}, we can equivalently replace constraint \eqref{eq:routing} (which includes both complicating and simple variables)  with the following one (which includes the complicating variables only):
\begin{align}\label{eq:const:delay2_e}
    z_{sen} \leq x_{ns} l_{en},~\forall s,e,n.
\end{align}

Now, considering the new form of the delay constraints as stated in \eqref{eq:const:delay2_a}-\eqref{eq:const:delay2_d}, 
the reformulated target problem can be written as 
\begin{subequations}\label{eq: Org Prob with G}
\begin{align} 
 \mathcal P_{org}^{\text{reformulated}} : ~ &\max_{ \mathbf{X} , \mathbf{Z} }
\max_{ \mathbf{Y}  ,  \mathbf{U} } ~ \mathcal{U}_{\text{total}} \left( \mathbf{X}  ,\mathbf{Y},\mathbf{U} \right) \\ 
    \text{subject to} \,\, ~
    &\eqref{eq:capacity:storage}, \eqref{eq:const:delay2_e}, \\
    &\mathbf{G} \left( \mathbf{X} , \mathbf{Z} ,\mathbf{Y},\mathbf{U} \right) \geq \mathbf{0}, \\
    &\left( \mathbf{X} , \mathbf{Z} \right) \in D_{X,Z},\\
    & \left( \mathbf{Y} ,  \mathbf{U} \right) \in \mathcal D_{Y,U},
\end{align}
\end{subequations}
where 
$D_{X,Z} \triangleq \{0,1\}^{(E+1)\times S} \times \{0,1\}^{S \times E \times  (E+1)}$ defines the joint domain of service placement and auxiliary request transfer decision variables, and $D_{Y,U} \triangleq 
    [0,1]^{ S \times E \times (E+1) } \times \mathbb R_+^{(E+1) \times S}$ defines the joint domain of request scheduling and CPU sizing decision variables.   
    Finally,  $\mathbf{G} \left( \mathbf{X} , \mathbf{Z} , \mathbf{Y},\mathbf{U} \right)$ is 
the vector of constraint functions defined on  $
D_{X,Z} \times \mathcal D_{Y, U}$ and  determined by all the constraints that include the simple variables only or the constraints that include both the simple and complicating variables
,  i.e., 
 \eqref{eq:sum1}, \eqref{eq:capacity:bw}, \eqref{eq:capacity:mu}, \eqref{eq:stability:node}, \eqref{eq:stability:link}, and \eqref{eq:const:delay2_a}--\eqref{eq:const:delay2_d}.

Next, using primal decomposition \cite{palomar2006tutorial}, 
we decompose the above problem into an equivalent form of master and inner sub-problems that are over complicating and simple variables, respectively. The inner problem is simply the original problem for some fixed values of the complicating variables $ \left( \hat{\mathbf{X}}, \hat{\mathbf{Z}} \right) \in \mathcal D_{X,Z} $:  
\begin{align}  \label{eq: P_in}
\mathcal{P}_{in}:
\quad f \left( \hat{\mathbf{X}} , \hat{\mathbf{Z}}  \right) \equiv &\max_{ (\mathbf{Y} , \mathbf{U}) \in \mathcal D_{Y,U} } ~ \mathcal R_{\text{srv}}({\mathbf{Y}}) - \mathcal P_{\text{cpu}}({\mathbf{U}}) -\mathcal P_{\text{trn}}({\mathbf{Y}})\notag \\
    &\text{subject to}\,\, ~ \mathbf{G} \left( \mathbf{\hat X}, \hat{\mathbf{Z}} , \mathbf{Y},\mathbf{U} \right) \geq \mathbf{0}, 
\end{align}
and the outer problem, a.k.a. the master problem, is defined as 
\begin{align}
\mathcal{P}_{out}:
\quad    &\max_{ \left( \mathbf{X}, \mathbf{Z} \right) \in \mathcal \mathcal D_{X,Z} \cap V } v ( \mathbf{X}, \mathbf{Z} ) - \mathcal P_{\text{str}}({\mathbf{X}})  \\
    \text{subject to} \notag \\
    &\eqref{eq:capacity:storage}, \eqref{eq:const:delay2_e}, \notag 
\end{align}
where $ v( \mathbf{X}, \mathbf{Z} ) = \{\sup_{ \mathbf{Y},\mathbf{U} \in \mathcal D_{Y,U} } \mathcal R_{\text{srv}}({\mathbf{Y}})  -\mathcal P_{\text{cpu}}({\mathbf{U}}) -\mathcal P_{\text{trn}}({\mathbf{Y}})  ; ~\mathbf{G} \left( \mathbf{X}, \mathbf{Z} , \mathbf{Y},\mathbf{U} \right) \geq \mathbf{0} \}  $ and  $V\triangleq \{ \left( \mathbf{X}, \mathbf{Z} \right) : \mathbf{G} \left( \mathbf{\hat X}, \hat{\mathbf{Z}} , \mathbf{Y},\mathbf{U} \right) \geq \mathbf{0} \, ~ \text{for some}\, \left( \mathbf{Y},\mathbf{U} \right) \in \mathcal D_{Y,U} \}$, indicates the set of complicating variables for which there is a feasible solution to the original problem. 
%
As such, $ \mathcal D_{X,Z} \cap V $ can be thought of as the projection of the feasible region of 
 onto $\left( \mathbf{X}, \mathbf{Z} \right)$-space.



Note that according to primal decomposition \cite{boyd2004convex}, the above decomposition is optimal and the master problem is equivalent to the original problem in \eqref{eq: Org Prob with G}.  
Moreover, it can be easily verified that the inner problem is convex, and hence, can be solved in polynomial time using the conventional convex optimization approaches such as primal-dual interior point methods   \cite{boyd2009convex}.  
 However, the  function $ v ( \mathbf{X}, \mathbf{Z} ) $ and the set $V$ in the definition of the master problem are only known implicitly via their definitions. To deal with this challenge, in the following, we employ a cutting-plane approach to build adequate representations of the extremal value of the inner problem as a function of the complicating integer variables and the set of values of the integer variables for which the inner problem is feasible. 
For this purpose, the dual representations of $v$ and $V$ are adopted, and a relaxed version of the master problem is constructed in an iterative manner, which maintains an upper-bound to the solution of the original problem \cite{geoffrion1972generalized}, as follows:
\begin{subequations}
\label{eq: P relaxed master}
\begin{align}
    \mathcal{P}_{out}^{{\text{relaxed}}}:  
    & \max_{ \textsf{UB}, \left(  \mathbf{X}, \mathbf{Z}  \right) \in \mathcal D_{X,Z} }  \textsf{UB}  \\ 
    \text{subject to}& \notag \\
    & \textsf{UB} \leq L^\ast \left(  \mathbf{X}, \mathbf{Z}  ; \mathbf{\mu}^{k_1} \right) , ~ \forall {k_1} \in \left\{ 1,\ldots,\tau^{\text{feas}} \right\} , \label{eq: feas const} \\ 
    & L_\ast \left(\mathbf{X}, \mathbf{Z} ; \lambda ^{k_2} \right) \geq 0, ~ \forall {k_2} \in \left\{1,\ldots,\tau^{\text{infeas}} \right\}, \label{eq: infeas const}
\end{align}
\end{subequations}
where $ L^\ast \left(  \mathbf{X}, \mathbf{Z}  ; \mathbf{\mu}^{k_1} \right) $ and $ L_\ast \left(\mathbf{X}, \mathbf{Z} ; \lambda ^{k_2} \right) $ are Lagrange functions defined as 
\begin{align} \label{eq: L11}
 L^\ast \left(  \mathbf{X}, \mathbf{Z}  ; \mathbf{\mu} \right) \triangleq \sup_{ \mathbf{Y},\mathbf{U} \in \mathcal D_{Y,U}  }  \left\{ \mathcal{U}_{\text{total}}  \left( \mathbf{X}, \mathbf{Y},\mathbf{U} \right) + {\mathbf{\mu}}^T \mathbf{G} \left( \mathbf{X},  \mathbf{Z} , \mathbf{Y},\mathbf{U}  \right)  \right\},
\end{align}
where $\mathbf{\mu}$ 
is Lagrangian multiplier vectors corresponding to solving the inner problem $ \mathcal{P}_{in} $ when it is feasible, and
\begin{align} \label{eq: L21}
    & L_\ast \left(\mathbf{X}, \mathbf{Z} ; \mathbf{\lambda} \right)  \triangleq \sup_{ \mathbf{Y},\mathbf{U} \in \mathcal D_{Y,U} }  \left\{  {\mathbf{\lambda}}^T \mathbf{G} \left( \mathbf{X},  \mathbf{Z} , \mathbf{Y},\mathbf{U}  \right)   \right\},  
\end{align} 
where $ \mathbf{\lambda} $ is any vector satisfying 
\begin{align}\label{eq: feas lambda set}
\left\lbrace \mathbf{\lambda} \geq \mathbf{0},  \mathbf{\lambda}^T \mathbf{1} =1 , {\mathbf{\lambda}}^T \mathbf{G} \left( \mathbf{X},  \mathbf{Z} , \mathbf{Y},\mathbf{U}  \right) < 0 \right\rbrace . 
\end{align}
In the proposed method, the above relaxed master problem is iteratively constructed, in which,   $\tau^{\text{feas}}$ and $\tau^{\text{infeas}}$ are the counters for 
 the number of iterations in which the inner problem is feasible and infeasible, respectively.  
In fact, at each iteration, one constraint of the form of \eqref{eq: feas const}, i.e., \textit{optimality cut}, or  \eqref{eq: infeas const}, i.e., \textit{feasibility cut}, will be added to the constraints  of  the  relaxed  master  problem. 
Accordingly, after each iteration, $\tau^{\text{feas}}$ or $\tau^{\text{infeas}}$ will be incremented by one,
 depending 
 on the feasibility of the inner problem at that iteration. 

\begin{challenge} \label{ch: obtaining Lagrange func.s} 
It should be noted that solving the relaxed master problem in \eqref{eq: P relaxed master}  is still challenging. This is mainly due to the fact that obtaining the Lagrange functions still involves optimization over both simple and complicating variables. 
\end{challenge} 

\subsection{The Proposed 
Solution Strategy}\label{sec:solution:GBD}

First, to tackle  Challenge \ref{ch: obtaining Lagrange func.s}, in the following, we derive  explicit forms for the Lagrange functions in \eqref{eq: L11} and \eqref{eq: L21}, i.e., as explicit functions of $\left( \hat{\mathbf{X}} , \hat{\mathbf{Z}}  \right)$, resulting in bypassing the involved  optimization problems. 
For this purpose, first, note that for any fixed $\left( \hat{\mathbf{X}} , \hat{\mathbf{Z}}  \right)$ where the inner problem is feasible, the upper-bounding function defined in \eqref{eq: L11} can be written as 
\begin{align} \label{eq: L upper star}
 \hspace{-4 pt}   L^*\left ( \hat{\mathbf{X}} , \hat{\mathbf{Z}}  ; \mathbf{\hat \mu} \right) =  \max_{(\mathbf{Y}, \mathbf{U}) \in \mathcal D_{Y, U}}[~ & \mathcal R_{\text{srv}}({\mathbf{Y}}) - \mathcal P_{\text{cpu}}({\mathbf{U}}) -\mathcal P_{\text{trn}}({\mathbf{Y}})   \\
    & + \mathbf{\hat \mu}^T G ( \hat{\mathbf{X}} , \hat{\mathbf{Z}}  , \mathbf{Y},\mathbf{U} ) ~]
    - \mathcal P_{\text{str}} ( \hat{\mathbf{X}} )  , \notag
\end{align}
in which $\mathbf{\hat \mu}$ is the optimal Lagrange multiplier vector  obtained in solving the inner problem $\mathcal{P}_{in}$. 
Moreover, 
due to convexity of the inner problem $\mathcal{P}_{in}$ and Slater's condition (which is easy to verify in this problem) \cite{boyd2004convex}, 
the optimal solution to the above optimization equals to the optimal solution of the inner problem $\mathcal{P}_{in}$. 
Furthermore, according  to  complementary slackness \cite{boyd2004convex} 
for the inner problem, in the objective function \eqref{eq: L upper star}, the terms in  $\mathbf{\hat \mu}^T G( \hat{\mathbf{X}} , \hat{\mathbf{Z}} ,\mathbf{U})$ associated with the constraints including only the continuous variables are equal to zero. As such,   
only the terms associated with constraints 
\eqref{eq:const:delay2_a}--\eqref{eq:const:delay2_d} 
remain in the objective function \eqref{eq: L upper star}. 

According to the above discussions, if we let $\hat f$ denote the optimal value and $(\mathbf{\hat Y}, \mathbf{\hat U})$ denote the optimal solution of the inner problem, we have 
\begin{align} \label{eq: L super index star}
    L^* & \left(\mathbf{X}, \mathbf{Z}, \mathbf{\hat \mu} \right) =  \hat{f} - \mathcal P_{\text{str}}({\mathbf{X}})   \notag\\ 
    &+ \mathbf{\hat \mu_{\ref{eq:const:delay2_a}}}^T ~ [  D^{\text{max}}_s (\hat u_{ns}/\nu_s - \hat \Lambda_{ns}) - z_{sen}, \forall s, e, n=e ] \notag \\
    &+ \mathbf{\hat \mu_{\ref{eq:const:delay2_b}}}^T ~ [  (1-\alpha) D^{\text{max}}_s (\hat u_{ns}/\nu_s - \hat \Lambda_{ns}) - z_{sen}, \forall s, e, n\neq e ] \notag \\
    &+ \mathbf{\hat \mu_{\ref{eq:const:delay2_c}}}^T ~ [  [\alpha D^{\text{max}}_s - d^{\text{prp}}_{en}]^+ (B^l_{en} - \hat \Lambda^l_{en}) - z_{sen},  \forall s, e, n \neq e ] \notag \\
    &+ \mathbf{\hat \mu_{\ref{eq:const:delay2_d}}}^T ~ [ z_{sen} - \hat y_{sen},  \forall s,e,n ].
\end{align}

On the other hand, for any complicating variables $ ( \mathbf{X}, \mathbf{Z} ) $ 
where the inner problem is infeasible, 
first a feasibility problem defined in \eqref{eq: feas lambda set} should be solved to obtain $\mathbf{\lambda}$ and then, 
the 
objective function $ L_\ast $ in \eqref{eq: L21} can be obtained. 
For this purpose, first note that if the inner problem is infeasible, then there must exist $\mathbf{\lambda}$ satisfying \eqref{eq: feas lambda set} such that $ \max_{(\mathbf{Y}, \mathbf{U}) \in \mathcal D_{Y, U}} \left \{  {\mathbf{\lambda}}^T \mathbf{G} \left( \mathbf{X},  \mathbf{Z} , \mathbf{Y},\mathbf{U}  \right) \right \} < 0  $ (see \cite{geoffrion1972generalized} for details). 
To obtain such a $\mathbf{\lambda}$, we can  solve the following problem (by primal-dual optimization method) 
which derives  an equivalent form for  Problem  \eqref{eq: L21}: 
\begin{align}\label{eq: P_feas}
&\min_{ \mathbf{\lambda} \geq \mathbf{0} } \max_{(\mathbf{Y}, \mathbf{U}) \in \mathcal D_{Y, U}} \{ \mathbf{\lambda}^T \mathbf{G} \left( \mathbf{X}, \mathbf{Z}, \mathbf{Y}, \mathbf{U} \right)  
    \text{s.t.} ~  \mathbf{\lambda}^T \mathbf{1} =1  \}   ,  \notag \\
&=\max_{\gamma} \min_{ \mathbf{\lambda} \geq \mathbf{0} } \max_{(\mathbf{Y}, \mathbf{U}) \in \mathcal D_{Y, U}} \{ \mathbf{\lambda}^T \mathbf{G} \left( \mathbf{X}, \mathbf{Z}, \mathbf{Y}, \mathbf{U} \right) + \gamma ( \mathbf{\lambda}^T \mathbf{1} - 1 )  \}  , \notag \\
&=\max_{\gamma} \min_{ \mathbf{\lambda} \geq \mathbf{0} } \max_{(\mathbf{Y}, \mathbf{U}) \in \mathcal D_{Y, U}} \{ - \gamma + \mathbf{\lambda}^T ( \mathbf{G} \left( \mathbf{X}, \mathbf{Z}, \mathbf{Y}, \mathbf{U} \right) + \gamma \mathbf{1} ) \}  , \notag \\   
&=\min_{\gamma, (\mathbf{Y}, \mathbf{U}) \in \mathcal D_{Y, U} }  \{ \gamma \}  ~~\text{s.t.} ~  \mathbf{G} \left( \mathbf{X}, \mathbf{Z}, \mathbf{Y}, \mathbf{U} \right) + \gamma \mathbf{1} \geq 0 .   
\end{align}
%
The first equality holds due to convexity of the set $ \left \{ \mathbf{\lambda} \geq \mathbf{0}: ~ \mathbf{\lambda}^T \mathbf{1}=1 \right \} $, and the third equality holds due to the convexity of $ \mathbf{G} \left( \mathbf{X}, \mathbf{Z}, \mathbf{Y}, \mathbf{U} \right)  $ with respect to $ \left( \mathbf{Y}, \mathbf{U} \right) $. Then, since $ \mathcal D_{Y, U} $ is a convex set, the solution to this problem is easily obtained (e.g., using conventional convex optimization methods such as primal-dual interior point method \cite{boyd2004convex}) and its associated multipliers compose the desired $ \mathbf{\lambda} $.
We refer to the optimal solution of this problem and its associated Lagrange multipliers as $(\bar{\mathbf{Y}}, \bar{\mathbf{U}})$ and $\hat{\mathbf{\lambda}}$, respectively, and obtain the feasibility function defined in \eqref{eq: L21} as
\begin{align}\label{eq: L subscript star}
    L_* & \left( \mathbf{X}, \mathbf{Z}  ; \mathbf{\hat \lambda} \right) = 
\hat{\mathbf{\lambda}}^T \mathbf{G} \left( \mathbf{X}, \mathbf{Z}, \bar{\mathbf{Y}}, \bar{\mathbf{U}} \right)     
\end{align}

Now, using the explicit forms of $ L^\ast $ and $ L_\ast $ as functions of $ \left( \mathbf{X} , \mathbf{Z} \right) $, as obtained in \eqref{eq: L super index star}, and \eqref{eq: L subscript star} and substituting them in \eqref{eq: L11} and \eqref{eq: L21}, respectively, we can easily construct and solve the relaxed master problem \eqref{eq: P relaxed master}. 
Note that, as aforementioned, the relaxed master problem maintains an upper-bound to the solution of the original problem. In addition, clearly, the inner problem 
 maintains a lower-bound to the solution of the original problem. 
 In the following, based on generalized Benders decomposition \cite{geoffrion1972generalized}, we iteratively select and add new $ \mathbf{\mu}^{k_1} $  and $ \mathbf{\lambda}^{k_2} $, and propose an iterative algorithm which 
 proceeds to reduce the gap between these two bounds, until it becomes small enough and hence, the convergence to the optimal solution is obtained. 

The proposed GBD-based algorithm can be described as follows:

\textbf{Step 1 (Initialization):} The algorithm starts with a feasible initial solution for the complicating variables, i.e., the service placement and request transfer variables. 

\textbf{Step 2 (Solving the Inner Problem):} For the fixed complicating variables, it uses the primal-dual interior-point method to solve the inner problem \eqref{eq: P_in} and find the optimal request scheduling and CPU sizing variables and the corresponding vector of Lagrange multipliers. 
As solution of the inner problem and the latest solution of the relaxed master problem maintain a lower-bound ($\mathsf{LB}$) and an upper-bound ($\mathsf{UB}$) to the solution of the original problem, respectively, one of the following three cases occurs:

\textbf{Step 2A:} If the inner problem is feasible and the gap between its solution and the upper-bound is still large, then the upper-bound will be recalculated and a new optimality cut in the form of \eqref{eq: L super index star} 
is determined and added to the relaxed master problem in \eqref{eq: P relaxed master}. 

\textbf{Step 2B:} If the inner problem is infeasible, then 
a new feasibility cut in the form of 
\eqref{eq: L subscript star} is determined and added to the relaxed master problem \eqref{eq: P relaxed master}. 

\textbf{Step 2C:} If the inner problem is feasible and the gap between its solution and the upper-bound is small enough, then the algorithm stops and the latest solutions of the inner and relaxed master problems are returned as the output of the algorithm.

\textbf{Step 3 (Solving the Relaxed Master Problem):}
Having the new feasibility or optimality cuts added in Step~2 of the algorithm, the relaxed master problem is updated and solved  by any convex optimization algorithm to realize a new solution for the service placement and service transfer variables, as well as a new upper-bound to the solution of the original problem. Then, the algorithm returns to Step 2.

The above steps are done iteratively until the gap between the lower and upper bounds becomes small enough. Finally, it outputs the optimal values of the service placement, service transfer, request scheduling, and CPU sizing variables. Algorithm~\eqref{alg:AO_method} presents the pseudo-code of the proposed algorithm.  

\begin{algorithm}[t]
\caption{ The Proposed GBD-Based Service Placement, Request Routing, and CPU Sizing Algorithm.}
\label{alg:AO_method}
\begin{algorithmic}[1]
\State \textbf{Initialization:} Pick  $\left( \hat{\mathbf{X}} , \hat{\mathbf{Z}} \right) \in 
\mathcal{D}_{X,Z}  $, and set  $\tau^{\text{feas}}=0$, $\tau^{\text{infeas}}=0$, $\textsf{UB}= + \infty$, $\textsf{LB}= - \infty $, and the convergence tolerance parameter $\epsilon \geq 0$.

\State \textbf{Inner problem:} Solve the inner problem $\mathcal{P}_{in}$ in \eqref{eq: P_in} for $\left( \hat{\mathbf{X}} , \hat{\mathbf{Z}} \right)$, 
by the primal-dual interior point method, 
and obtain (if exist) its optimal value  $f\left( \hat{\mathbf{X}} , \hat{\mathbf{Z}} \right)$, its optimal solution $\left( \hat{\mathbf{Y}},\hat{\mathbf{U}} \right)$ and the corresponding Lagrange  multipliers vector $\mathbf{\hat \mu}$. \label{alg:inner}

\If{$ - \infty < f\left( \hat{\mathbf{X}} , \hat{\mathbf{Z}} \right) < \mathsf{UB} - \epsilon $ }


\State $ \tau^{\text{feas}} \leftarrow \tau^{\text{feas}} +1$.
\State $\mathbf{\hat \mu}^{\tau^{\text{feas}}} \leftarrow \mathbf{\hat \mu} $
\State $\textsf{LB} \leftarrow \max \{\textsf{LB},  f\left( \hat{\mathbf{X}} , \hat{\mathbf{Z}} \right)\}$.
\EndIf

\If{$ f\left( \hat{\mathbf{X}} , \hat{\mathbf{Z}} \right) = - \infty $} 

\State $ \tau^{\text{infeas}} \leftarrow \tau^{\text{infeas}} +1$.

\State Solve the 
problem in \eqref{eq: P_feas} for $\left( \hat{\mathbf{X}} , \hat{\mathbf{Z}} \right)$, and obtain its optimal solution $\left( \bar{\mathbf{Y}},\bar{\mathbf{U}} \right)$ and the corresponding Lagrange  multipliers vector $ \hat{\mathbf{\lambda}}^{\tau^{\text{infeas}}}$.

\EndIf

\If{
$ \mathsf{UB} - \epsilon \leq f\left( \hat{\mathbf{X}} , \hat{\mathbf{Z}} \right) < \infty $}
\State Go to Line \ref{alg:output}. 

\EndIf

\State \textbf{Outer problem:} Solve the  current relaxed master problem  defined by \eqref{eq: P relaxed master},   
where functions $ L^\ast $ and $ L_\ast $ functions are defined in \eqref{eq: L super index star} and \eqref{eq: L subscript star}, respectively, 
by any 
convex optimization algorithm, and obtain its optimal solution $ \left( 
\mathsf{UB}, \hat{\mathbf{X}} , \hat{\mathbf{Z}} \right) $. \label{alg:master}





\State Return to Line \ref{alg:inner}.

\State {\textbf{Output:}} $ \left( \mathbf{X}^\ast , \mathbf{Z}^\ast, \mathbf{Y}^\ast , \mathbf{U}^\ast \right) \leftarrow \left( \hat{\mathbf{X}} ,  \hat{\mathbf{Z}} , \hat{\mathbf{Y}}, \hat{\mathbf{U}} \right) $. \label{alg:output}
\end{algorithmic}
\end{algorithm}

A few remarks regarding the proposed algorithm are in place. 

\begin{remark}[\textbf{Complexity of the Inner and Outer Problems}]

    Note that solving the inner and the relaxed master problems in the proposed algorithm are much easier than solving the original problem. 
    This is due to the fact that for fixed values of the complicating variables $\left( \mathbf{X},\mathbf{Z} \right)$, problem $\mathcal{P}_{in}$ is a convex optimization problem, and hence, can be solved in polynomial time by using the conventional convex optimization approaches (e.g., primal-dual interior point methods) \cite{boyd2009convex}. Moreover, it is noted that using primal-dual interior point methods, the optimal Lagrange multipliers corresponding to the inner problem can also be achieved simultaneously, which will be used in the algorithm as well.
    In addition, regarding the relaxed master problem \eqref{eq: P relaxed master}, 
    the optimizations involved in determining the Lagrange functions $L^\ast$ and $L_\ast$ are also convex optimization problems. This is  due to the fact that their objective functions are linearly separable in the binary and real variables, and hence, 
the involved supremums over $\mathcal D_{Y,U}$ can be taken independently of the binary variables $\left( \mathbf{X} , \mathbf{Z} \right)$. 
    As such, the relaxed master problem \eqref{eq: P relaxed master} is no longer an MINLP, but just an ILP that can be solved efficiently via the existing efficient approaches. 
    
\end{remark}

\section{Convergence Analysis}\label{sec:convergence_analysis}

In this section, we present theoretical  convergence results of the proposed algorithm. In particular, we prove the proposed algorithm converges to 
the global optimal solution of the reformulated target problem \eqref{eq: Org Prob with G} in a finite number of iterations. For this purpose, we first state and prove the following lemmas.

\begin{lemma}\label{lem: lem1_alternative}
    The set $\mathcal D_{Y,U}$ is non-empty, convex, bounded and closed. 
\end{lemma}

\begin{proof}
 The proof follows a direct result of the definition of the set $\mathcal D_{Y,U}$, and we omit it due to page limit. 
\end{proof}

\begin{lemma}
    The constraint functions in $ \mathbf{G} \left( \mathbf{X} , \mathbf{Z} ,\mathbf{Y},\mathbf{U} \right) $ are concave and continuous on $\mathcal D_{Y,U}$, for any fixed $\left( \mathbf{X} , \mathbf{Z} \right) \in \mathcal D_{X,Z} $.
\end{lemma}

\begin{proof}
    For any fixed value of $\left( \mathbf{X} , \mathbf{Z} \right) \in \mathcal D_{X,Z} $, all the constraint functions in $ \mathbf{G} \left( \mathbf{X} , \mathbf{Z} ,\mathbf{Y},\mathbf{U} \right) $ are linear over $\left( \mathbf{Y},\mathbf{U} \right)$ and hence, also concave. Moreover, due to \ref{lem: lem1_alternative} and the linearity of the constraint functions in $ \mathbf{G} \left( \mathbf{X} , \mathbf{Z} ,\mathbf{Y},\mathbf{U} \right) $, these functions are also continuous on $D_{Y,U}$, for each fixed $ \left( \mathbf{X} , \mathbf{Z} \right) \in \mathcal{D}_{X,Z} $. 
\end{proof}

\begin{lemma}\label{lem: strong duality}
Strong duality holds for the  inner problem $\mathcal{P}_{in}$. 

\end{lemma}\label{lem: slater}

\begin{proof}
It can be easily verified that the inner problem $\mathcal{P}_{in}$ in \eqref{eq: P_in} is convex. 
Furthermore, since for any feasible 
$ \left( \mathbf{X} , \mathbf{Z} \right)  $, the inner problem includes only linear constraints and its domain is non-empty, it  satisfies Slater's condition too.  Consequently, strong duality holds between $\mathcal P_{in}$ and its dual. 
\end{proof}

\begin{lemma}\label{lem: opt multplr vec}

For any fixed $\left( \mathbf{X} , \mathbf{Z} \right) \in \mathcal D_{X,Z} \cap V $, the solution of the inner problem $\mathcal{P}_{in}$ 
  is either unbounded or bounded 
  with the inner problem possessing an optimal multiplier vector.
\end{lemma}

\begin{proof}
First, note that according to the definition of $V$, for any $\left( \mathbf{X} , \mathbf{Z} \right) \in \mathcal D_{X,Z} \cap V $, the problem $\mathcal{P}_{in}$ never gets infeasible. Therefore, either $f \left( \mathbf{X} , \mathbf{Z}  \right) = + \infty $ or $f \left( \mathbf{X} , \mathbf{Z}  \right) $ is finite. In the latter case, since according to Lemma \ref{lem: strong duality}, strong duality holds between  problem $\mathcal{P}_{in} $ and its dual,  
this problem possesses an optimal 
multiplier vector  which is the Lagrangian multiplier vector   associated with the constraints in $\mathcal P_{in}$ (which is the solution to the dual problem) and is obtainable 
via the primal-dual interior point methods. 



\end{proof}

\begin{theorem}[Finite Convergence to the Global Optimal Solution]
 The proposed 
 method in Algorithm  \ref{alg:AO_method}  terminates in a finite number of steps for any given $\epsilon \geq 0$, therefore,  converging  to the global optimal solution of the reformulated target problem in \eqref{eq: Org Prob with G}. 
\end{theorem}

\begin{proof}

According to Lemmas \ref{lem: lem1_alternative}--\ref{lem: opt multplr vec} and using \cite{floudas1995nonlinear}[Theorem 6.3.4], 
the proposed GBD-based algorithm terminates in a finite number of iterations for any given $\epsilon \geq 0$. Consequently, when $ \epsilon $ is set to $ 0 $, then the upper-bound and the lower-bound obtained from the algorithm for the optimal value of Problem \eqref{eq: Org Prob with G} equal to each other and hence, the algorithm converges to the optimal solution for the target problem. 
\end{proof}

%% file: part5.tex
\section{Numerical Performance Evaluation}\label{sec:evaluation}
In this section, we perform numerical experiments to show the performance of the proposed 
GBD-based approach in terms of the SP's profit, admission rate, end-to-end-delay, cloud load, and convergence rate, 
and compare it to several baselines. The proposed solution and baseline schemes are implemented using CVXPY \cite{diamond2016cvxpy}. Moreover, in the proposed solution, the inner convex problem is solved using MOSEK \cite{mosek} and the master ILP problem is solved using Gurobi \cite{gurobi}. 

\textbf{Baselines.} For comparison, we also implement the following baselines in different evaluation scenarios. For fair comparison, in all benchmarks, the ILP problem is solved using Gurobi and the inner problem is solved by MOSEK. All experiments are run on a server with $ 630 $ GB RAM and $ 16 $ Cores with Ubuntu OS. 
\begin{enumerate}
\item {\underline {Non-Cooperative MEC (NC)}}:  
All service requests at the BSs are served locally, i.e., $z_{sen}=0,\, \forall s, e, n \neq e$.  Under this assumption, the inner and master problems can 
be solved independently for each $\mathrm{BS}_e$.
\item {\underline{Alternating Optimization (AO)}}:  
At each iteration, the optimization problem in \eqref{eq: Org Prob with G} 
is solved with respect to the set of real and integer variables, alternatively, assuming the other set of variables are fixed. 
\item {\underline{CSPR \cite{somesula2023cooperative}}}: This 
baseline first relaxes all the integer variables, and then employs 
an approximation scheme based on relaxation and rounding techniques to solve the cooperative MEC optimization problem. 
More details can be found in \cite{somesula2023cooperative} and are omitted here due to space limitations. 
\item {\underline{SPR3 \cite{poularakis2020service}}}: This baseline leverages a randomized rounding technique to propose an approximation algorithm that 
provably achieves approximation guarantees while violating the resource constraints in a bounded way. More details can be found in \cite{poularakis2020service}. 
\end{enumerate}

\begin{table}[!t]
\centering
  \caption{Services.}
  \setlength{\tabcolsep}{0.9pt} 
\begin{tabular}{ l c c c c c 
 }
  \hline
 & Storage \hspace{3 pt}  & Computation \hspace{3 pt}  & Input/output Size
 \\ 
 Description &$\sigma_s$&$\nu_s$&$\beta_s$
 \\ 
 & (GB) &(M CPU cyles)& (Mb) 
  \\\hline
 Face Recognition &$[2,10]$&$[0.375,3]$&$[1,8]$
 \\
 GZip (compression) &$0.02$&$[0.04,0.32]$&$[1,6]$
 \\
 Augmented Reality&$[2,20]$&$ [0.375,3] $&$[1,6]$
 \\
  Video Streaming &$[1,10]$&$0.0001$&$[1,25]$
  \\
\end{tabular}
 \label{table:services}
\end{table}

\textbf{Network and SP Setup.}
We consider $E=10$ BSs located uniformly at random in a $ 500 \times 500 ~\text{m}^2 $  region. 
The storage, computation and communication capacities of each $ \text{BS}_e $ is set to $C_e=100 $~GBs, $M_e=20$~GHz, and $B_e=100$~Mbps, respectively. 
We assume that a communication link 
with propagation latency of $d^{\text{prp}}_{ee'} = 0.5\mu s, \forall e, e'$ 
exists between each two neighbouring BSs that are less than $150$~m apart. Moreover, the propagation delay between BSs and the cloud is set to $ 50 $~ms \cite{tran2018cooperative}. 


We consider $S=3000$ heterogeneous  services uniformly distributed from four real latency-sensitive service categories, namely Video streaming (VS), Face recognition (FR), Gzip (compression) and Augmented reality (AR), with specifications that can be found in Table~\ref{table:services}.  
The target delay of all these services are also set to  $D^{\text{max}}_s = 400$~ms. The values of this table are inline with the real service specifications (e.g., see \cite{ SmartFace, HPReality}). Moreover, 
we assume $\lambda_{e}=10$ rqts/s. Moreover, total requests for each service category at each $\mathrm{BS}_e$ is distributed according to Zipf with $\alpha=0.8$, i.e, $\lambda_{es} \propto f_{es}^{-\alpha}$ where $f_{es}$ is the rank of service $s$ among all services of its correspondent category at $\mathrm{BS}_e$. 

\textbf{Leasing costs of storage, computation, and bandwidth.} An estimate of bandwidth, storage and computation costs derives from the rates charged by cloud providers like Google \cite{googleCloud}. GB download prices convert to a rate of around $\$0.2$ and $\$0.5$ per Mbps of hour traffic for inter BS and transit traffic to the cloud.
Likewise, 
 a rate of around $0.2$ per GB per month and $0.4$ per GHz/h is the current offer for standard access storage (HDD) and general-purpose machine types, respectively. We assume that these costs at each $\mathrm{BS}_e$ is scaled by $\omega_e>1$, which is drawn uniformly from $ \left[ 4, 7 \right] $.

\textbf{Revenue.} We consider weighted linear revenue such that $R^{\text{srv}}_{ns}(\Lambda_{ns})=w_{es} \Lambda_{ns}$, where $w_{es}$ denotes the the  hourly revenue obtained by serving one request of  service $s$ per second at server $n$. Different weights allow prioritizing the admission of requests for one service or one region over the others depending on the reliability requirements. In our evaluations, the weights are chosen service-dependent and 
proportional to the amount of input/output size of each service, such that $w_{ns}$ is equal to $3 \beta_s$. 

\subsection{Performance Comparison of the Proposed Solution and the Baselines}

In the rest of this section, we evaluate the performance of the proposed method and compare it to the illustrated baselines, in terms of various performance metrics.  

First, Fig. \ref{fig:Util} illustrates the SP's profit versus 
the number of BSs. It can be verified from this figure that the proposed method significantly outperforms all the baselines in maximizing the net profit of the SP in performing cooperative MEC in the network.  In addition, when the number of the BSs increases, 
the  SP's profit in our method rapidly increases (almost linearly), owing to the increased cooperation provided by the proposed scheme among the BSs.

\begin{figure}[t]
  \centering
    \includegraphics[width=0.53\textwidth]{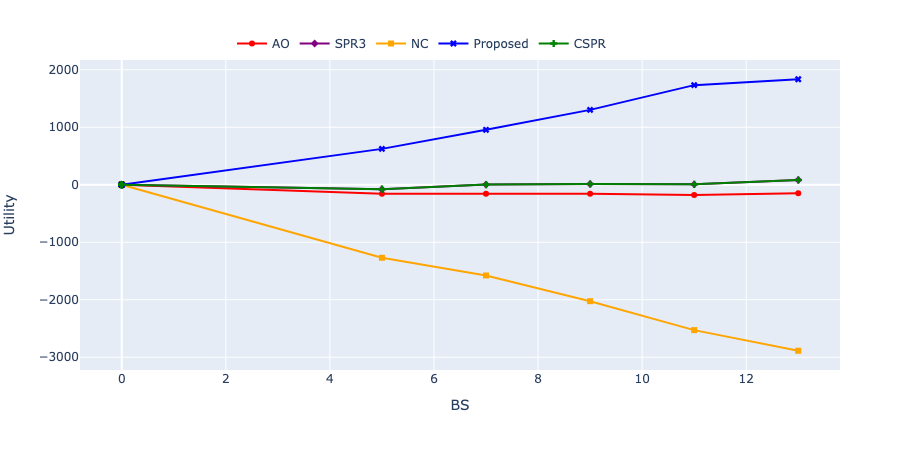} 
      \caption{Comparison of the profit of the SP in the proposed scheme and the baselines.}
      \label{fig:Util}
\end{figure}

Next, Fig. \ref{fig:time} depicts the running time of the proposed method and different baselines. As can be seen from this figure, the proposed method demonstrates a comparable running time to those of the fastest baselines, while it reaches to the highest objective values within almost the same time. Therefore, while the proposed method is proved to converge to the optimal solution of the reformulated target problem, its convergence time is as almost low as the fastest heuristics. This result aligns with our theoretical analysis that the proposed method converges to the optimal point in finite number of iterations (of the master problem). Moreover, it also demonstrates the low computational complexity of the proposed scheme.

\begin{figure}[t]
  \centering
    \includegraphics[width=0.53\textwidth]{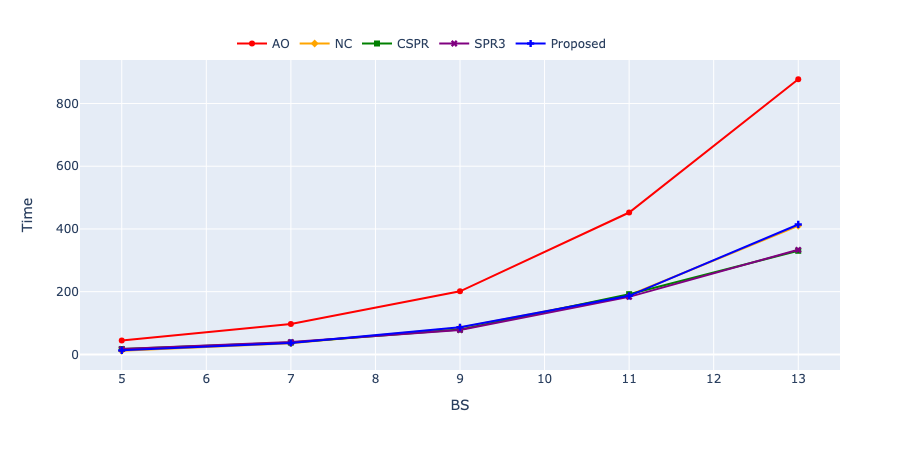} 
      \caption{Comparison of the running time of the proposed algorithm and the baselines.}
      \label{fig:time}
\end{figure}

Fig. \ref{fig:CHR} compares the cache hit performance of the proposed method to those of the baselines, for varying number of BSs. As can be verified from this figure, the proposed method can reach  significantly more cache hit ratios compared to the baselines. Therefore, the proposed method can serve  significantly more services compared to the state-of-the-art baselines. 

\begin{figure}[t]
  \centering
    \includegraphics[width=0.53\textwidth]{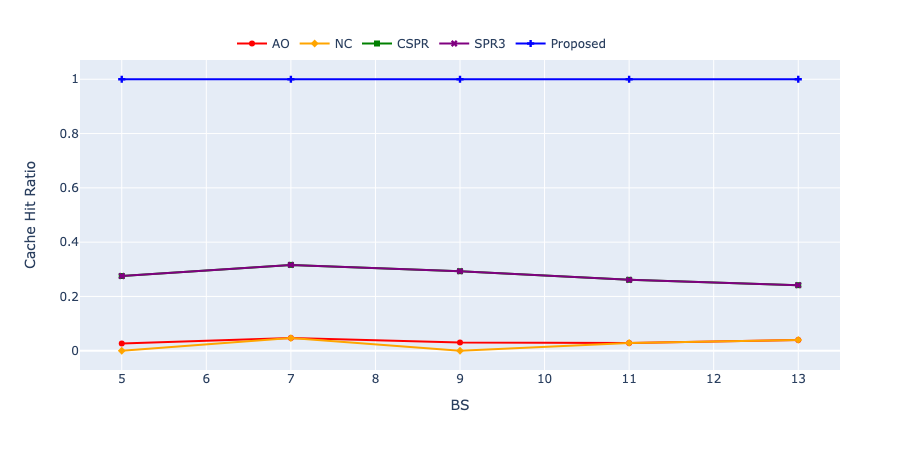} 
      \caption{Comparison of the cache hit ratio in the proposed scheme and the baselines.}
      \label{fig:CHR}
\end{figure}

Finally, Fig. \ref{fig:Del} shows the average end-to-end delay of the served services under the proposed scheme and the baselines. As can be seen, the proposed method maintains the end-to-end delay of the served services  under the considered target delay (which was $D^{\text{max}}_s = 400$~ms). Moreover, compared to the baselines, its average end-to-end delay is almost minimum. 

\begin{figure}[t]
  \centering
    \includegraphics[width=0.53\textwidth]{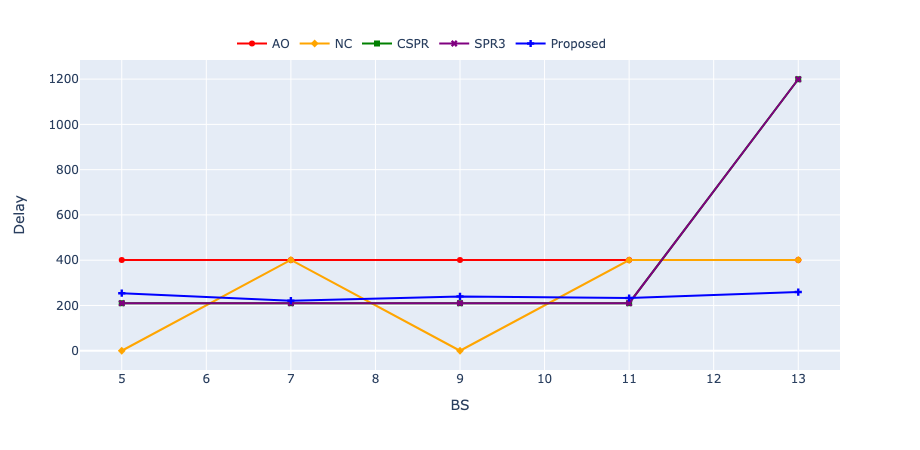} 
      \caption{Comparison of the average end-to-end delay in the proposed scheme and the baselines.}
      \label{fig:Del}
\end{figure}